\documentclass{llncs}
\usepackage{makeidx}  
\usepackage{amsmath,amscd,amsfonts,amssymb}
\usepackage{algorithmic}
\usepackage{algorithm}
\usepackage{url}

\usepackage{tikz}
\usepackage{verbatim}
\usepackage{relsize}
\usetikzlibrary{trees,snakes}
\tikzset{fontscale/.style = {font=\relsize{#1}}}

\newcommand{\Z}{\mathbb{Z}} 
\newcommand{\Q}{\mathbb{Q}} 
\newcommand{\F}{\mathbb{F}} 
\newcommand{\g}{\mathfrak{g}} 
\newcommand{\lk}{\mathfrak{l}} 

\newcommand{\Cl}{\operatorname{Cl}}

\newcommand{\OO}{\mathcal{O}}
\newcommand{\OK}{{\mathcal{O}_K}}

\newcommand{\ag}{\mathfrak{a}} 
\newcommand{\bg}{\mathfrak{b}} 
\newcommand{\pg}{\mathfrak{p}} 
\newcommand{\Nm}{\mathcal{N}}

\newcommand{\disc}{\operatorname{disc}}

\newcommand{\End}{\operatorname{End}}

\newcommand{\ket}[1]{|#1\rangle}

\newtheorem{heuristic}{Heuristic}
\newtheorem{corrolary}{Corrolary}

\begin{document}
\mainmatter              
\title{A note on the security of CSIDH}
\titlerunning{On the security of CSIDH}  
%
\author{Jean-Fran\c{c}ois Biasse\inst{1} \and
Annamaria Iezzi\inst{1} \and Michael J. Jacobson, Jr. \inst{2}}
%
%
%
\institute{
Department of Mathematics and Statistics \\
University of South Florida\\
\email{\{biasse,aiezzi\}@usf.edu}
\and Department of Computer Science \\
University of Calgary\\
\email{jacobs@ucalgary.ca}
}

\maketitle              

\begin{abstract}
We propose an algorithm for computing an isogeny between two elliptic curves $E_1,E_2$ 
defined over a finite field such that there is an imaginary quadratic order $\OO$ 
satisfying $\OO\simeq \End(E_i)$ for $i = 1,2$. This concerns ordinary curves and 
supersingular curves defined over $\F_p$ (the latter used in the recent CSIDH proposal). Our 
algorithm has heuristic asymptotic run time $e^{O\left(\sqrt{\log(|\Delta|)}\right)}$ and 
requires polynomial quantum memory and $e^{O\left(\sqrt{\log(|\Delta|)}\right)}$ classical memory, 
where $\Delta$ is the discriminant of 
$\OO$. This asymptotic complexity outperforms all other available method for computing isogenies. 

We also show that a variant of our method has asymptotic run time 
$e^{\tilde{O}\left(\sqrt{\log(|\Delta|)}\right)}$ while requesting only polynomial memory 
(both quantum and classical). 

\end{abstract}
\section{Introduction}

Given two elliptic curves $E_1,E_2$ defined over a finite field $\F_q$, the isogeny problem 
is to compute an isogeny $\phi:E_1\rightarrow E_2$,  i.e.\ a non-constant morphism that maps the identity point on $E_1$ to the identity point on $E_2$. 
There are two different types of elliptic curves: ordinary and supersingular. 
The latter have very particular properties that impact the resolution of the isogeny problem. 
The first instance of a cryptosystem based on the hardness of computing isogenies was due to Couveignes~\cite{couveignes}, and 
its concept was independently rediscovered by Stolbunov~\cite{Stolbunov}. Both proposals used ordinary curves. 

Childs, Jao and Soukharev observed in~\cite{jao_souk_quantum} that the problem 
of finding an isogeny between two ordinary curves $E_1,E_2$ defined over $\F_q$ and having the same endomorphism ring could be reduced 
to the problem of finding a subgroup of the dihedral group. More 
specifically, let $K = \Q(\sqrt{t^2 - 4q})$ where $t$ is the trace of the Frobenius endomorphism of the curves, 
and let $\OO\subseteq K$ be the quadratic order isomorphic to the ring of 
endomorphisms of $E_1$ and $E_2$. Let $\Cl(\OO)$ be the 
ideal class group of $\OO$. Classes of ideals act on isomorphism classes of curves with endomorphism 
ring isomorphic to $\OO$. Finding an isogeny between $E_1$ and $E_2$ boils down to finding an ideal $\ag$ 
such that $[\ag]* \overline{E}_1 = \overline{E}_2$ where $*$ is the action of $\Cl(\OO)$, $[\ag]$ is the 
class of $\ag$ in $\Cl(\OO)$ and $\overline{E}_i$ is 
the isomorphism class of the curve $E_i$. Childs, Jao and Soukharev showed that this could be done by finding a subgroup 
of  $\Z_2 \ltimes \Z_N$, where $N = \#\Cl(\OO)\sim \sqrt{|t^2 - 4q|}$. Using Kuperberg's sieve~\cite{Kuperberg2}, this requires 
$2^{O\left(\sqrt{\log(N)}\right)}$ queries to an oracle that computes the action of the class of an ideal $\ag$. 
Childs et al. used a method with complexity in $2^{\tilde{O}\left(\sqrt{\log(N)}\right)}$ to evaluate this oracle, meaning that the total cost is 
$2^{\tilde{O}\left(\sqrt{\log(N)}\right)}$. 

To avoid this subexponential attack, Jao and De Feo~\cite{jao_defeo_short} described an analogue of these 
isogeny-based systems that works with supersingular curves. The endomorphism ring of such curves is a maximal 
order in a quaternion algebra. The non-commutativity of the (left)-ideals acting on isomorphism classes of curves 
thwarts the attacks mentioned above, but it also restricts the possibilities offered by supersingular isogenies, which 
are typically used for a Diffie-Hellman type of key exchange (known as SIDH) and for digital signatures. Most recently, 
two independent works revisited isogeny-based cryptosystems by restricting themselves to cases where the subexponential 
attacks based on the action of $\Cl(\OO)$ was applicable. The scheme known as CSIDH by Castryck et al.~\cite{CSIDH} 
uses supersingular curves and isogenies defined over $\F_p$, while the scheme of De Feo, Kieffer and Smith~\cite{ordinary_graphs} 
uses ordinary curves with many practical optimizations. In both cases, the appeal of using 
commutative structures is to allow more functionalities, such as static-static key exchange protocols that are not 
possible with SIDH without an expensive Fujisaki-Okamoto transform~\cite{SIDH_static}. 
The downside is that one needs to carefully assess the hardness of the computation of isogenies in 
this context. 

\paragraph{\textbf{Contribution.}} Let $E_1,E_2$ be two elliptic curves 
defined over a finite field such that there is an imaginary quadratic order $\OO$ 
satisfying $\OO\simeq \End(E_i)$ for $i = 1,2$ and $\Delta = \disc(\OO)$. 
In this note, we provide new insight 
regarding the security of CSIDH as follows:
\begin{enumerate}
\item We describe a quantum algorithm for computing an isogeny between $E_1$ and $E_2$ with heuristic 
 asymptotic run time in $e^{O\left(\sqrt{\log(|\Delta|)}\right)}$ and 
with quantum memory in $\operatorname{Poly}\left(\log(|\Delta|)\right)$ and quantumly accessible classical 
memory in $e^{O\left(\sqrt{\log(|\Delta|)}\right)}$. 
\item We show that we can use a variant of this method to compute an isogeny between $E_1$ and $E_2$ 
in time  $e^{\left(\frac{1}{\sqrt{2}}+o(1)\right)\sqrt{\ln(|\Delta|)\ln\ln(|\Delta|)}}$ with polynomial memory (both classical 
and quantum).
\end{enumerate}
Our contributions bear similarities to the recent independent work of 
Bonnetain and Schrottenloher~\cite{Bonnetain}. 
The main differences are that they rely on a generating set $\lk_1,\ldots,\lk_u$ where $u\in \Theta(\log(|\Delta|))$ 
of the class group provided with the CSIDH protocol, and 
that their approach does not have asymptotic 
run time in $e^{O\left(\sqrt{\log(|\Delta|)}\right)}$. The asymptotic 
run time of the method of~\cite{Bonnetain} was not analyzed. However, 
using similar techniques as the ones presented in this paper, one could 
conclude that the run time of the method of~\cite{Bonnetain} is in $e^{\tilde{O}\left(\sqrt{\log(|\Delta|)}\right)}$. 
The extra logarithmic factors in the exponent come from the use of the BKZ reduction~\cite{BKZ} 
with block size in $O\left(\sqrt{\log(|\Delta|)}\right)$ in a lattice of $\Z^u$. Section~\ref{sec:quantum_oracle} 
elaborates on the differences between our algorithm and~\cite{Bonnetain}. The run time of the variant described in 
Contribution~2 is asymptotically comparable to that of the algorithm of Childs, Jao and Soukharev~\cite{jao_souk_quantum}, 
and to that of Bonnetain and Schrottenloher~\cite{Bonnetain} (if its exact time complexity was to be worked out). 
The main appeal of our variant is the fact that it uses a polynomial amount of memory, which is likely to impact the 
performances in practice. 


\section{Mathematical background} 

An elliptic curve defined over a finite field $\F_q$ of characteristic $p\neq
2,3$ is an algebraic variety given by an equation of the form $E:\ y^2 = x^3 + ax + b$, 
where $a$, $b \in \mathbb F_q$ and  $4a^3 + 27b^2 \neq 0$. A more general form gives an
affine model in the case $p=2,3$ but it is not useful in the scope of
this paper since we derive an asymptotic result. The set of points of
an elliptic curve can be equipped with an additive group law. Details
about the arithmetic of elliptic curves can be found in many
references, such as~\cite[Chap. 3]{silverman_elliptic_curves}.

Let $E_1,E_2$ be two elliptic curves defined over $\F_q$. An isogeny
$\phi\colon E_1\to E_2$ is a non-constant rational map defined over
$\F_q$ which is also a group homomorphism from $E_1$ to $E_2$. Two
curves are isogenous over $\F_q$ if and only if they have the same
number of points over $\F_q$ (see~\cite{tate}). Two curves over $\F_q$
are said to be isomorphic over $\overline{\F}_q$ if there is an
$\overline{\F}_q$-isomorphism between their group of points. Two such
curves have the same $j$-invariant given by $j:= 1728\frac{4a^ 3}{4a^
  3 + 27b^ 2}$.  In this paper, we treat isogenies as mappings between
(representatives of) $\overline{\F}_q$-isomorphism classes of elliptic
curves. In other words, given two $j$-invariants $j_1,j_2\in\F_q$, we
wish to construct an isogeny between (any) two elliptic curves
$E_1,E_2$ over $\F_q$ having $j$-invariant $j_1$ (respectively
$j_2$). Such an isogeny exists if and only if $\Phi_\ell(j_1,j_2)=0$
for some $\ell$, where $\Phi_\ell(X,Y)$ is the $\ell$-th modular
polynomial.

Let $E$ be an elliptic curve defined over $\F_q$. An isogeny between
$E$ and itself defined over $\F_{q^n}$ for some $n>0$ is called an
endomorphism of $E$.  The set of endomorphisms of $E$ is a ring that
we denote by $\End(E)$. For each integer $m$, the multiplication-by-$m$ map $[m]$ on $E$
is an endomorphism. Therefore, we always have $\Z\subseteq \End(E)$.
Moreover, to each isogeny $\phi\colon E_1\to E_2$ corresponds an
isogeny $\widehat{\phi}\colon E_2\to E_1$ called its dual
isogeny. It satisfies $\phi\circ\widehat{\phi}=[m]$ where $m =
\deg(\phi)$. For elliptic curves defined over a finite field, we know that
$\Z\subsetneq\End(E)$. In this particular case, $\End(E)$ is either an
order in an imaginary quadratic field (and has $\Z$-rank 2) or an
order in a quaternion algebra ramified at $p$ (the characteristic of the base field) and $\infty$ (and has
$\Z$-rank 4). In the former case, $E$ is said to be ordinary while in
the latter it is called supersingular. When a supersingular curve 
is defined over $\F_p$, 
then the ring of its $\F_p$-endomorphisms is isomorphic to an imaginary 
quadratic order, much like in the ordinary case. 

An order $\OO$ in a field $K$ such that $[K:\Q]=n$ is a subring of $K$
which is a $\Z$-module of rank $n$. The notion of ideal of $\OO$ can
be generalized to fractional ideals, which are sets of the form $\ag =
\frac{1}{d}I$ where $I$ is an ideal of $\OO$ and $d\in\Z_{>0}$. The
invertible fractional ideals form a multiplicative group
$\mathcal{I}$, having a subgroup consisting of the invertible
principal ideals $\mathcal{P}$. The ideal class group $\Cl(\OO)$ is by
definition $\Cl(\OO):= \mathcal{I}/\mathcal{P}$. In $\Cl(\OO)$, we
identify two fractional ideals $\ag,\bg$ if there is $\alpha\in K$
such that $\bg = (\alpha)\ag$. We denote by $[\ag]$ the 
class of the fractional ideal $\ag$ in $\Cl(\OO)$. 
The ideal class group is finite and its
cardinality is called the class number $h_\OO$ of $\OO$. For a
quadratic order $\OO$, the class number satisfies $h_\OO \leq
\sqrt{|\Delta|}\ln(|\Delta|)$, where $\Delta$ is the discriminant of $\OO$.

The endomorphism ring of an elliptic curve plays a crucial role in
most algorithms for computing isogenies between curves. The class
group of $\End(E)$ acts transitively on isomorphism classes of
elliptic curves (that is, on $j$-invariants of curves) having the same
endomorphism ring. More precisely, the class of an ideal
$\ag\subseteq\OO$ acts on the isomorphism class of a curve $E$ with
$\End(E)\simeq \OO$ via an isogeny of degree $\Nm(\ag)$ (the algebraic
norm of $\ag$).  Likewise, each isogeny $\varphi\colon E\to E'$ where
$\End(E)=\End(E')\simeq \OO$ corresponds (up to isomorphism) to the
class of an ideal in $\OO$.  From an ideal $\ag$ and the $\ell$-torsion
(where $\ell=\Nm(\ag)$), one can recover the kernel of $\varphi$, and
then using V\'{e}lu's formulae~\cite{velu}, one can derive the
corresponding isogeny. We denote by $[\ag]*\overline{E}$ the action of the 
ideal class of $\ag$ on the isomorphism class of the curve $E$. 
The typical strategy to evaluate the action 
of $[\ag]$ is to decompose it as a product of classes of prime ideals of 
small norm, and evaluate the action of each prime ideals as $\ell$-isogenies. 
This strategy was first described by Couveignes~\cite{couveignes} and later by Br\"{o}ker-Charles-Lauter~\cite{BCL} 
and reused in many subsequent works.
\medskip

\noindent\textbf{Notation: } In this paper, $\log$ denotes the base 2 logarithm 
while $\ln$ denotes the natural logarithm. 

\section{The CSIDH static-static key exchange}

As pointed out in~\cite{GPST16}, the original SIDH key agreement 
protocol is not secure when using the same secret over multiple 
instances of the protocol. This can be fixed by a Fujisaki-Okamoto transform~\cite{SIDH_static} 
at the cost of a drastic loss of performance, requiring additional points in the protocol, and loss of flexibility, for example, the inability to reuse keys.  These issues motivated the 
description of CSIDH~\cite{CSIDH} which uses supersingular curves defined 
over $\F_p$. 

When Alice and Bob wish to create a shared secret, they rely on their 
long-term secrets $[\ag]$ and $[\bg]$ which are classes of ideals in 
the ideal class group of $\OO$, where $\OO$ is isomorphic to the 
$\F_p$-endomorphism ring of a supersingular curve $E$ defined over 
$\F_p$. Much like the original Diffie-Hellman protocol~\cite{DH}, 
Alice and Bob proceed as follow: 
\begin{itemize}
 \item Alice sends $[\ag]*\overline{E}$ to Bob. 
 \item Bob sends $[\bg]*\overline{E}$ to Alice.
\end{itemize}
Then Alice and Bob can separately recover their shared secret
$$[\ag\bg]*\overline{E} = [\bg]*[\ag]*\overline{E} = [\ag]*[\bg]*\overline{E}.$$
The existence of a quantum subexponential attack forces the use of larger keys 
for the same level of security, which is partly compensated by the fact that 
elements are represented in $\F_p$, and are thus more compact.  Recommended parameter sizes and attack costs 
from~\cite{CSIDH} for 80, 128, and 256 bit security are listed in Table~\ref{tab:ab:SIDH_security}.

\begin{table}[ht]
\caption{\label{tab:ab:SIDH_security} Claimed security of CSIDH~\cite[Table 1]{CSIDH}.}
\begin{center}
\begin{tabular}{|c|c||c|c|}
\hline
NIST & $\log(p)$ & Cost quantum attack  & Cost classical attack  \\
\hline
1 & 512  &  $2^{62}$  &  $2^{128}$  \\ 
3 & 1024  &  $2^{94}$  &  $2^{256}$  \\ 
5 & 1792  & $2^{129}$   &  $2^{448}$ \\
\hline
\end{tabular}
\end{center}
\end{table}

\section{Asymptotic complexity of isogeny computation}\label{sec:quantum_space}

In this section, we show how to combine the general framework for computing isogenies between 
curves whose endomorphism ring is isomorphic to a quadratic order (due to Childs, Jao and Soukharev~\cite{jao_souk_quantum} in the 
ordinary case and to Biasse Jao and Sankar in the supersingular case~\cite{INDOCRYPT_14}) with 
the efficient evaluation of the class group action of Biasse, Fieker and Jacobson~\cite{ANTS_XII} to produce a quantum algorithm 
that finds an isogeny between $E_1,E_2$. We give two variants of our method: 
\begin{itemize}
 \item Heuristic time complexity $2^{O\left(\log(|\Delta|)\right)}$, polynomial quantum memory and 
 quantumly accessible classical memory in $2^{O\left(\log(|\Delta|)\right)}$.
 \item Heuristic time complexity $e^{\left(\frac{1}{\sqrt{2}}+o(1)\right)\sqrt{\ln(|\Delta|)\ln\ln(|\Delta|)}}$ with polynomial memory (both classical 
and quantum). 
\end{itemize}

\subsection{Isogenies from solutions to the Hidden Subgroup Problem}

As shown in~\cite{INDOCRYPT_14,jao_souk_quantum}, the computation of an isogeny between 
$E_1$ and $E_2$ such that there is an imaginary quadratic order with $\OO\simeq \End(E_i)$ for 
$i = 1,2$ can be done by exploiting the action of the ideal class group of $\OO$ on isomorphism 
classes of curves with endomorphism ring isomorphic to $\OO$. In particular, this concerns the 
cases of 
\begin{itemize}
 \item ordinary curves, and 
 \item supersingular curves defined over $\F_p$. 
\end{itemize}

Assume we are looking for $\ag$ such that $[\ag]*\overline{E}_1 = \overline{E}_2$. Let $A = \Z/d_1\Z\times \cdots \times \Z/d_k\Z\simeq \Cl(\OO)$ 
be the elementary decomposition of $\Cl(\OO)$. 
Then we define 
a quantum oracle on $\Z/2\Z\ltimes A$ by 
\begin{equation}\label{eq:HSP_oracle_A}
 f(x,y):= \left\{ \begin{array}{ll}
         \ket{[\ag_y]*\overline{E}_1}  & \mbox{if $x=0$},\\
        \mbox{$\ket{[\ag_y]*\overline{E}_2}$} & \mbox{if $x =1$},\end{array} \right.
\end{equation}
where $[\ag_y]$ is the element of $\Cl(\OO)$ corresponding to $y\in A$ via the isomorphism $\Cl(\OO)\simeq A$. 
Let $H$ be the subgroup of $\Z/2\Z\ltimes A$ of the periods of $f$. This means that $f(x,y) = f(x',y')$ if and 
only if $(x,y)-(x',y')\in H$. Then $H = \{(0,0) , (1 , s)\}$ where $s\in A$ such that 
$[\ag_s]*\overline{E}_1 = \overline{E}_2$. The computation of $s$ can thus be done through the 
resolution of the Hidden Subgroup Problem in $\Z/2\Z\ltimes A$. In~\cite[Appendix A]{jao_souk_quantum}, Childs, Jao and 
Soukharev generalized the subexponential-time polynomial space dihedral HSP algorithm of Regev~\cite{regev} to the 
case of an arbitrary Abelian group $A$. Its run time is in $e^{\left(\sqrt{2}+o(1)\right)\sqrt{\ln(|A|)\ln\ln(|A|)}}$ 
with a polynomial memory requirement. 
Kuperberg~\cite{Kuperberg2} describes a family of algorithms, one of which has running time in $e^{O\left(\sqrt{\log(|A|)}\right)}$ 
while requiring polynomial quantum memory and $e^{O\left(\sqrt{\log(|A|)}\right)}$ quantumly accessible 
classical memory. The high-level approach for finding an isogeny from the dihedral HSP 
is skteched in Algorithm~\ref{alg:HSP}.

 \begin{algorithm}[ht]
   \caption{Quantum algorithm for finding the action in $\Cl(\OO)$}
\begin{algorithmic}[1]\label{alg:HSP}
\REQUIRE Elliptic curves $E_1,E_2$, imaginary quadratic order $\OO$ such that $\End(E_i)\simeq \OO$ for $i=1,2$
such that there is $[\ag]\in\Cl(\OO)$ satisfying 
$[\ag]*\overline{E}_1 = \overline{E}_2$.
\ENSURE $[\ag]$
\STATE Compute $A = \Z/d_1\Z\times \cdots \times\Z/d_k \Z$ such that $A\simeq \Cl(\OO)$. 
\STATE Find $H =  \{(0,0) , (1 , s)\}$ by solving the HSP in $\Z/2\Z\ltimes A$ with oracle~\eqref{eq:HSP_oracle_A}.
\RETURN $[\ag_s]$
\end{algorithmic}
\end{algorithm}

\begin{proposition}\label{prop:HSP_regev}
Let $N = \#\Cl(\OO)\sim \sqrt{|\Delta|}$. Algorithm~\ref{alg:HSP} is correct and requires:
\begin{itemize}
 \item $e^{O\left( \sqrt{\log(N)}\right)}$ queries 
to the oracle defined by~\eqref{eq:HSP_oracle_A} 
while requiring a
$\operatorname{Poly}(\log(N))$ quantum memory overhead and $e^{O\left( \sqrt{\log(N)}\right)}$ quantumly accessible classical memory overhead when 
using Kuperberg's second dihedral HSP algorithm~\cite{Kuperberg2} in Step~2. 
\item $e^{\left(\sqrt{2}+o(1)\right)\sqrt{\ln(|N|)\ln\ln(|N|)}}$ queries to the oracle defined by~\eqref{eq:HSP_oracle_A} while requiring only 
polynomial memory overhead when using the dihedral HSP method ofusing~\cite[Appendix A]{jao_souk_quantum} in Step~2. 
\end{itemize}
\end{proposition}

\begin{remark}
Algorithm~\ref{alg:HSP} only returns the ideal class $[\ag]$ whose action on $\overline{E}_1$ gives us $\overline{E}_2$. 
This is all we are interested in as far as the analysis of isogeny-based cryptosystems goes. However, this is 
not technically an isogeny between $E_1$ and $E_2$. We can use this ideal to derive an actual isogeny by evaluating 
the action of $[\ag]$ using the oracle of Section~\ref{sec:quantum_oracle} together with the method of~\cite[Alg. 4.1]{BCL}. 
This returns an isogeny $\phi:E_1\rightarrow E_2$ as a composition of isogenies of small degree $\phi = \prod_i \phi_i^{e_i}$ without 
increasing the time complexity. Also note that the output fits in polynomial space if the product is not evaluated, 
otherwise, it needs $2^{\tilde{O}\left( \sqrt[3]{\log(N)}\right)}$ memory. 
\end{remark}

\subsection{The quantum oracle}\label{sec:quantum_oracle}

To compute the oracle defined in~\eqref{eq:HSP_oracle_A}, Childs, Jao and Soukharev~\cite{jao_souk_quantum} used a purely classical 
subexponential method deriving from the general subexponential class group computation algorithm of Hafner and McCurley~\cite{hafner}.  
This approach, mentioned in \cite{CSIDH}, was first suggested by Couveignes \cite{couveignes}. In a recent independent 
work~\cite{Bonnetain}, 
Bonnetain and Schrottenloher used a method that bears similarities with our oracle described in 
this section. They combined a quantum algorithm for computing the class group with classical methods from Biasse, Fieker 
and Jacobson~\cite[Alg. 7]{ANTS_XII} for evaluating the action of $[\ag]$ with a precomputation of $\Cl(\OO)$. 
More specifically, let 
$\lk_1,\ldots,\lk_u$ be prime ideals used to create the long term secret $\ag$ of Alice. This means that there are (small) 
$(e_1,\ldots,e_u)\in\Z^u$ such that $\ag = \prod_i \lk_i^{e_i}$. Let $\mathcal{L}$ be the lattice of exponent vectors 
$(f_1,\ldots,f_u)$ such that $\prod_i \lk_i^{f_i} = (\alpha)$ for some $\alpha$ in the field of fractions of $\OO$. In other 
words, the ideal class $\left[\prod_i \lk_i^{f_i}\right]$ is the neutral element of $\Cl(\OO)$. The high-level approach used 
in~\cite{Bonnetain} deriving from~\cite[Alg. 7]{ANTS_XII} is the following: 
\begin{enumerate}
 \item Compute a basis $B$ for $\mathcal{L}$.
 \item Find a BKZ-reduced basis $B'$ of $\mathcal{L}$. 
 \item Find $(h_1,\ldots,h_u)\in\Z^u$ such that $[\ag] =\left[\prod_i\lk_i^{h_i}\right]$.
 \item Use Babai's nearest plane method on $B'$ to find short $(h'_1,\ldots,h'_u)\in\Z^u$ such that $[\ag] =\left[\prod_i\lk_i^{h'_i}\right]$.
 \item Evaluate the action of $\left[\prod_i\lk_i^{h'_i}\right]$ on $\overline{E}_1$ by applying the action of the $\lk_i$ for $i=1,\ldots,u$.  
\end{enumerate}
Steps~1 and~2 can be performed as a precomputation. Step~1 takes quantum polynomial time by using standard techniques for solving 
an instance of the Hidden Subgroup Problem in $\Z^u$ where $u$  
satisfies $p = 4l_1\cdots l_u - 1$ for small primes $l_1,\ldots,l_u$.

The oracle of Childs, Jao and Soukharev~\cite{jao_souk_quantum} has asymptotic time complexity in $2^{\tilde{O}\left( \sqrt{\log(|\Delta|)}\right)}$ 
and requires subexponential space due to the need for the storage of 
a relation matrix of subexponential dimension. The oracle of Bonnetain and Schrottenloher~\cite{Bonnetain} relies on 
BKZ~\cite{BKZ} lattice reduction 
in a lattice in $\Z^u$.
Typically, $u\in \Theta(\log(p)) = \Theta(\log(|\Delta|))$, since $\sum_{q\leq l}\log(q) \in \Theta(l)$. 
In addition to not having a proven space complexity bound, the complexity of BKZ cannot be in $e^{\tilde{O}\left( \sqrt{\log(|\Delta|)}\right)}$ 
unless the block size is at least in $\Theta\left(\sqrt{\log(|\Delta|)}\right)$, which forces the overall complexity to be at best  
in $e^{\tilde{O}\left( \sqrt{\log(|\Delta|)}\right)}$.

Our strategy differs from that of Bonnetain and Schrottenloher on the following points: 
\begin{itemize}
 \item Our algorithm does not require the basis $\lk_1,\ldots,\lk_u$ provided with CSIDH. 
 \item The complexity of our oracle is in $e^{\tilde{O}\left(\sqrt[3]{\log(|\Delta|)}\right)}$ 
 (instead of $e^{\tilde{O}\left(\sqrt{\log(|\Delta|)}\right)}$ for the 
 method of~\cite{Bonnetain}), thus leading to an overall 
 complexity of $e^{O\left(\sqrt{\log(|\Delta|)}\right)}$ (instead of $e^{\tilde{O}\left(\sqrt{\log(|\Delta|)}\right)}$ for the 
 method of~\cite{Bonnetain}).
 \item We specify the use of a variant of BKZ with a proven poly-space complexity.
\end{itemize}
To avoid the dependence on the parameter $u$, we need to rely on the heuristics stated by Biasse, Fieker and Jacobson~\cite{ANTS_XII} on the 
connectivity of the Caley graph of the ideal class group when a set of edges is $S = \{ \pg \in \operatorname{Poly}(\log(|\Delta|))\}$ 
with $\#S \leq \log(|\Delta|)^{2/3}$ where $\Delta$ is the discriminant of $\OO$. By assuming~\cite[Heuristic 2]{ANTS_XII}, we state that 
each class of $\Cl(\OO)$ has a representation over the class of ideals in $S$ with exponents less than $e^{\log^{1/3}(|\Delta|)}$. 
A quick calculation shows that there are asymptotically many more such products than ideal classes, but their distribution is not well 
enough understood to conclude that all classes decompose over $S$ with a small enough exponent vector. Numerical experiments reported 
in~\cite[Table 2]{ANTS_XII} showed that decompositions of random ideal classes over the first  $\log^{2/3}(|\Delta|)$ 
split primes always had exponents significantly less 
than $e^{\log^{1/3}(|\Delta|)}$. 
\begin{table}[ht]
\caption{\label{tab:heuristic} Maximal exponent occurring in short decompositions (over 1000 random elements of the class group). Table~2 of~\cite{ANTS_XII}.}
\begin{center}
\begin{tabular}{|c||c|c|c|}
\hline
$\log_{10}(|\Delta|)$ & $\log^{2/3}(|\Delta|)$  & Maximal coefficient & $e^{\log^{1/3}(|\Delta|)}$  \\
\hline
20  &  13  &  6  &  36 \\ 
25  &  15  &  8  &  48 \\ 
30  &  17  &  7  &  61 \\
35  &  19  &  9  &  75 \\ 
40  &  20  &  10  &  91 \\
45  &  22  &  14  &  110 \\
50  &  24  &  13  &  130 \\
\hline
\end{tabular}
\end{center}
\end{table}

\begin{heuristic}[With parameter $c > 0$]\label{heuristic}
Let $c > 0$ and $\OO$ an imaginary quadratic order of discriminant $\Delta$. Then there are  
 $(\pg_i)_{i\leq k}$ for $k = \log^{2/3}(|\Delta|)$ split prime ideals of norm less than $\log^c(|\Delta|)$ whose 
classes generate $\Cl(\OO)$. Furthermore, each class of $\Cl(\OO)$ has a representative of the form  
$\prod_i \pg_i^{n_i}$ for $|n_i|\leq e^{\log^{1/3}|\Delta|}$.
\end{heuristic}
A default choice for our set $S$ could be the first $\log^{2/3}(|\Delta|)$ split primes of $\OO$ (as in Table~\ref{tab:heuristic}). 
We can derive our results 
under the weaker assumption that the $\log^{2/3}(|\Delta|)$ primes generating the ideal class group do not have to be 
the first consecutive primes. Assume we know that $\Cl(\OO)$ is generated by at most $\log^{2/3}(|\Delta|)$ distinct classes 
of the split prime ideals of norm up to $\log^c(|\Delta|)$ for some constant $c > 0$. Our algorithm needs to first identify these 
prime ideals as they might not be the first consecutive primes. Let $\pg_1,\ldots,\pg_k$ be the prime ideals of norm up to 
$\log^c(|\Delta|)$. We first compute a basis for the lattice $\mathcal{L}$ of vectors $(e_1,\ldots,e_k)$ such that $\prod_i\pg_i^{e_i}$ 
is principal (in other words, the ideal class $\left[\prod_i\pg_i^{e_i}\right]$ is trivial). Let $M$ be the matrix whose rows are the 
vectors of a basis of $\mathcal{L}$. There is a polynomial time (and space) algorithm that finds a unimodular matrix $U$ such that 
\[ UM = H = 
\begin{bmatrix}
 h_{1,1}& 0      & \hdots & 0      \\
\vdots & h_{2,2}& \ddots & \vdots \\
\vdots & \vdots & \ddots & 0      \\
*      & *      & \hdots & h_{k,k}\\
\end{bmatrix},
\]
where $H$ is in Hermite Normal Form~\cite{Arne}. The matrix $H$ represents the unique basis of $\mathcal{L}$ such that $h_{i,i} > 0$, and $h_{j,j} > h_{i,j}$ 
for $i > j$. Every time $h_{i,i} = 1$, this means that we have a relation of the form $\left[\pg_i\right] = \left[ \prod_{j<i}\pg_j^{-h_{i,j}}\right]$. 
In other words, $\left[\pg_i\right]\in \langle [\pg_1],\ldots, [\pg_{i-1}]\rangle$. On the other hand, if $h_{i,i}\neq 1$, then 
$\left[\pg_i\right]\notin \langle [\pg_1],\ldots, [\pg_{i-1}]\rangle$. Our algorithm proceeds by computing the HNF of $M$, and every time 
$h_{i,i}\neq 1$, it inserts $\pg_i$ at the beginning of the list of primes, and moves the column $i$ to the first column, recomputes the HNF and 
iterates the process. In the end, the first $\log^{2/3}(|\Delta|)$ primes in the list generate $\Cl(\OO)$. 
 \begin{algorithm}[ht]
   \caption{Computation of  $\log^{2/3}(|\Delta|)$ primes that generate $\Cl(\OO)$}
\begin{algorithmic}[1]\label{alg:generating_set}
\REQUIRE Order $\OO$ of discriminant $\Delta$ and $c >0$. 
\ENSURE $\log^{2/3}(|\Delta|)$ split primes whose classes generate $\Cl(\OO)$. 
\STATE $S\leftarrow\{\text{Split primes $\pg_1,\dots,\pg_k$ of norm less than }\log^c(|\Delta|)\}$.
\STATE $\mathcal{L}\leftarrow$ lattice of vectors $(e_1,\ldots,e_k)$ such that $\prod_i\pg_i^{e_i}$ 
is principal.
\STATE Compute the matrix $H\in\Z^{k\times k}$ of a basis of $\mathcal{L}$ in HNF. 
\FOR{$j = k$ down to $\log^{2/3}(|\Delta|) + 1$}
\WHILE{$h_{j,j} \neq 1$}
\STATE Insert $\pg_j$ at the beginning of $S$. 
\STATE Insert the $j$-th column at the beginning of the list of columns of $H$. 
\STATE $H\leftarrow \text{HNF} (H)$. 
\ENDWHILE
\ENDFOR
\RETURN $\{\pg_1,\ldots,\pg_s\}$ for $s =  \log^{2/3}(|\Delta|)$.
\end{algorithmic}
\end{algorithm}

\begin{proposition}
Assuming Heuristic~\ref{heuristic} for the parameter $c$, Algorithm~\ref{alg:generating_set} is correct and runs in polynomial time 
in $\log(|\Delta|)$.
\end{proposition}

\begin{proof}
Step~2 can be done in quantum polynomial time with the $S$-unit algorithm of Biasse and Song~\cite{SODA16}. 
Assuming that $\log^{2/3}(|\Delta|)$ primes of norm less than $\log^c(|\Delta|)$ generate $\Cl(\OO)$, the loop of Steps~5 to~9 
is entered at most $j$ times as one of $[\pg_1],\ldots,[\pg_j]$ must be in the subgroup generated by the other $j-1$ ideal classes. 
The HNF computation runs in polynomial time, therefore the whole procedure runs in polynomial time. \qed
\end{proof}

Once we have $\pg_1,\ldots,\pg_s$, we compute a reduced basis $B'$ of the lattice 
$\mathcal{L}\subseteq \Z^s$ of the vectors $(e_1,\ldots,e_s)$ such that $\left[ \prod_i\pg_i^{e_i}\right]$ is trivial, and 
we compute the generators $\g_1,\ldots,\g_l$ such that $\Cl(\OO) = \langle\g_1\rangle\times\cdots\times\langle\g_l\rangle$ 
together with vectors $\vec{v}_i$ such that $\g_i = \prod_j \pg_j^{v_{i,j}}$. 

 \begin{algorithm}[ht]
   \caption{Precomputation for the oracle}
\begin{algorithmic}[1]\label{alg:precomputation}
\REQUIRE Order $\OO$ of discriminant $\Delta$ and $c >0$. 
\ENSURE Split prime ideals $\pg_1,\ldots,\pg_s$ whose classes generate $\Cl(\OO)$ where $s = \log^{2/3}(|\Delta|)$, reduced 
basis $B'$ of the lattice $\mathcal{L}$ of vectors $(e_1,\ldots,e_s)$ such that $\left[ \prod_i\pg_i^{e_i}\right]$ is trivial, 
generators $\g_1,\ldots,\g_l$ such that $\Cl(\OO) = \langle\g_1\rangle\times\cdots\times\langle\g_l\rangle$ and 
vectors $\vec{v}_i$ such that $\g_i = \prod_j \pg_j^{v_{i,j}}$.
\STATE $\pg_1,\ldots,\pg_s\leftarrow$ output of Algorithm~\ref{alg:generating_set}.
\STATE $\mathcal{L}\leftarrow$ lattice of vectors $(e_1,\ldots,e_s)$ such that $\prod_i\pg_i^{e_i}$ 
is principal.
\STATE Compute a BKZ-reduced matrix $B'\in\Z^{s\times s}$ of a basis of $\mathcal{L}$ with block size $\log^{1/3}(|\Delta|)$. 
\STATE Compute $U,V\in\operatorname{Gl}_s(\Z)$ such that $UB'V = \operatorname{diag}(d_1,\ldots,d_s)$ is the Smith Normal Form of $B'$. 
\STATE $l\leftarrow\min_{i\leq s}\{i\ \mid\ d_i\neq 1\}$. For $i\leq l$, $\vec{v_i}\leftarrow$ $i$-th column of $V$. 
\STATE $V'\leftarrow V^{-1}$. For $i\leq l$, $\g_i\leftarrow \prod_{j\leq s}\pg_j^{v'_{i,j}}$. 
\RETURN $\{\pg_1,\ldots,\pg_s\}$, $B'$, $\{\g_1,\ldots,\g_l\}$, $\{\vec{v}_1,\ldots,\vec{v}_l\}$. 
\end{algorithmic}
\end{algorithm}

\begin{lemma}
Let $\mathcal{L}$ be an $n$-dimentional lattice with input basis $B\in\Z^{n\times n}$, and let $\beta < n$ be a block size. Then 
the BKZ variant of~\cite{term_BKZ} used with Kannan's enumeration technique~\cite{Kannan_SVP} returns a basis $\vec{b'}_1,\ldots,\vec{b'}_n$ such that 
$$\|\vec{b'}_1\| \leq e^{\frac{n}{\beta}\ln(\beta)\left(1+o(1)\right)}\lambda_1\left(\mathcal{L}\right),$$ 
using time $\operatorname{Poly}(\operatorname{Size}(B))\beta^{\beta\left(\frac{1}{2e} + o(1)\right)}$ 
and polynomial space. 
\end{lemma}

\begin{proof}
According to~\cite[Th. 1]{term_BKZ}, $\|\vec{b'}_1\|\leq 4\left(\gamma_\beta\right)^{\frac{n-1}{\beta-1} + 3}\lambda_1\left(\mathcal{L}\right)$ where 
$\gamma_\beta$ is the Hermite constant in dimension $\beta$. As asymptotically $\gamma_\beta\leq \frac{1.744\beta}{2\pi e}(1+o(1))$ (see~\cite{bound_Hemite}), 
we get that $4\left(\gamma_\beta\right)^{\frac{n-1}{\beta-1} + 3} \leq e^{\frac{n}{\beta}\ln(\beta)\left(1+o(1)\right)}$. Moreover, this reduction is 
obtained with a number of calls to Kannan's algorithm that is bounded by $\operatorname{Poly}(n,\operatorname{Size}(B))$. According to~\cite[Th. 2]{stehle_enum}, 
each of these calls takes time $\operatorname{Poly}(n,\operatorname{Size}(B))\beta^{\beta\left(\frac{1}{2e} + o(1)\right)}$ and polynomial space, which 
terminates the proof.\qed
\end{proof}

\begin{proposition}
Assuming Heuristic~\ref{heuristic} for $c$, Algorithm~\ref{alg:precomputation} is correct, runs in time $e^{\tilde{O}\left(\sqrt[3]{\log(|\Delta|)}\right)}$ 
and has polynomial space complexity. 
\end{proposition}

The precomputation of Algorithm~\ref{alg:precomputation} allows us to design the quantum circuit that implements 
the function described in~\eqref{eq:HSP_oracle_A}. 
Generic techniques due to Bennett \cite{Ben89}
convert any algorithm taking time $T$ and space $S$
into a reversible algorithm taking time $T^{1+\epsilon}$ and space $O(S\log T)$. From a high level point of view, this is simply the 
adaptation of method of Biasse-Fieker-Jacobson~\cite[Alg. 7]{ANTS_XII} to the quantum setting. 

 \begin{algorithm}[ht]
   \caption{Quantum oracle for implementing $f$ defined in~\eqref{eq:HSP_oracle_A}}
\begin{algorithmic}[1]\label{alg:oracle}
\REQUIRE Curves $E_1,E_2$. 
Order $\OO$ of discriminant $\Delta$ such that $\End(E_i)\simeq \OO$ for $i=1,2$. 
Split prime ideals $\pg_1,\ldots,\pg_s$ whose classes generate $\Cl(\OO)$ where $s = \log^{2/3}(|\Delta|)$, reduced 
basis $B'$ of the lattice $\mathcal{L}$ of vectors $(e_1,\ldots,e_s)$ such that $\left[ \prod_i\pg_i^{e_i}\right]$ is trivial, 
generators $\g_1,\ldots,\g_l$ such that $\Cl(\OO) = \langle\g_1\rangle\times\cdots\times\langle\g_l\rangle$ and 
vectors $\vec{v}_i$ such that $\g_i = \prod_j \pg_j^{v_{i,j}}$. Ideal class $[\ag_y] \in\Cl(\OO)$ represented by the vector 
$\vec{y}=(y_1,\ldots,y_l)\in \Z/d_1\Z\times\cdots \times\Z/d_l\Z\simeq \Cl(\OO)$, and $x\in\Z/2\Z$. 
\ENSURE $f(x,\vec{y})$. 
\STATE $\vec{y}\leftarrow \sum_{i\leq l} y_i\vec{v_i}\in\Z^s$ (now $[\ag_y] = \left[ \prod_i\pg_i^{y_i}\right]$). 
\STATE Use Babai's nearest plane method with the  basis $B'$ to find $\vec{u}\in\mathcal{L}$ close to $\vec{y}$. 
\STATE $\vec{y}\leftarrow \vec{y} - \vec{u}$. 
\STATE \textbf{If} $x=0$ \textbf{then} $\overline{E}\leftarrow \overline{E}_1$ \textbf{else} $\overline{E}\leftarrow \overline{E}_2$.
\FOR{$i\leq s$}
\FOR{$j\leq y_i$}
\STATE $\overline{E}\leftarrow [\pg_i]*\overline{E}$. 
\ENDFOR
\ENDFOR
\RETURN $\ket{\overline{E}}$. 
\end{algorithmic}
\end{algorithm}

\begin{proposition}
Assuming Heuristic~\ref{heuristic} for $c$, Algorithm~\ref{alg:oracle} is correct and runs in quantum time $e^{\tilde{O}\left(\sqrt[3]{\log(|\Delta|)}\right)}$ and has polynomial 
space complexity. 
\end{proposition}

\begin{proof}
Each group action of Step~7 is polynomial in $\log(p)$ and in $\Nm(\pg_i)$. Moreover, following the arguments from the 
Biasse-Fieker-Jacobson method~\cite[Prop. 6.2]{ANTS_XII}, the $y_i$ are in $e^{\tilde{O}\left(\sqrt[3]{\log(|\Delta|)}\right)}$, which 
is the cost of Steps~5 to~9. The main observation allowing us to reduce the search to a close vector to the computation of 
a BKZ-reduced basis is that Heuristic~\ref{heuristic} gives us the promise that there is $\vec{u}\in\mathcal{L}$ at distance less 
than $e^{\sqrt[3]{\log(|\Delta|)}(1+o(1))}$ from $\vec{y}$.\qed
\end{proof}

\begin{corrolary}
Let $E_1,E_2$ be two elliptic curves and $\OO$ be an imaginary quadratic 
order of discriminant $\Delta$ such that $\End(E_i)\simeq \OO$ for $i=1,2$. 
Then assuming Heuristic~\ref{heuristic} for some constant $c>0$, there is a quantum algorithm for computing $[\ag]$ such that 
$[\ag]*\overline{E}_1 = \overline{E}_2$ with: 
\begin{itemize}
 \item heuristic time complexity $e^{O\left(\sqrt{\log(|\Delta|)}\right)}$, polynomial quantum memory and $e^{O\left(\sqrt{\log(|\Delta|)}\right)}$ 
 quantumly accessible classical memory,
 \item heuristic time complexity $e^{\left(\frac{1}{\sqrt{2}}+o(1)\right)\sqrt{\ln(|\Delta|)\ln\ln(|\Delta|)}}$ with polynomial memory (both classical 
and quantum). 
\end{itemize}
\end{corrolary}

\begin{remark}
Heuristic~\ref{heuristic} may be relaxed in the proof of the $e^{O\left(\sqrt{\log(|\Delta|)}\right)}$ asymptotic run time. As long as a number $k$ in $\tilde{O}\left( \log^{1-\varepsilon}(|\Delta|)\right)$ of prime ideals of polynomial norm generate the ideal class group and that each class has at least one decomposition involving exponents bounded 
by $e^{\tilde{O}\left( \log^{1/2 - \varepsilon}(|\Delta|)\right)}$, the result still holds by BKZ-reducing with block size 
$\beta = \sqrt{k}$. We refered to Heuristic~\ref{heuristic} as it was already present in the previous work of Biasse, Fieker and Jacobson~\cite{ANTS_XII}. 

For the poly-space variant, these conditions can be relaxed even further. It is known under GRH that a 
number $k$ in $\tilde{O}\left( \log(|\Delta|)\right)$ 
of prime ideals of norm less than $12\log^2(|\Delta|)$ 
generate the ideal class group. We only need to argue that 
each class can be decomposed with exponents bounded 
by $e^{\tilde{O}\left(\sqrt{\log(|\Delta|)}\right)}$. 
Then by using the oracle of Algorithm~\ref{alg:oracle} 
with block size $\beta = \sqrt{k}$, we get 
a run time of $e^{\tilde{O}\left(\sqrt{\log(|\Delta|)}\right)}$ with 
a poly-space requirement. 
\end{remark}

\section{A Remark on Subgroups}
It is well-known that the cost of quantum and classical attacks on isogeny based cryptosystems is more accurately measured in the size of the subgroup generated by the ideal classes used in the cryptosystem.  As stated in \cite{CSIDH}, in order to ensure that this is sufficiently large with high probability, 
the class number $N$ must have a large divisor $M$, meaning that there is a subgroup of order $M$ in the class group.  Assuming the Cohen-Lenstra heuristics, Hamdy and Saidak \cite[Sec.\ 3]{HamdySaidak_smoothness} prove that the smoothness probability for class numbers is essentially the same as for random integers of the same size.  Thus, for randomly-selected CSIDH system parameters, we expect that the class number will have a large prime divisor.  It is an open problem as to whether subgroups can be exploited to reduce the security of CSIDH, but there are nevertheless some minor considerations that can be taken into account to minimize the risk.

Constructing system parameters for which the class number has a large known divisor could be done by a quantum adversary using the polynomial-time algorithm to compute the class group and trial-and-error.  Using classical computation, it is in most cases infeasible.  Known methods to construct discriminants for which the class number has a known divisor $M$ use a classical result of Nagell \cite{Nagell_divisibility} relating the problem to finding discriminants $\Delta = c^2 D$ that satisfy $c^2 D = a^2 - 4 b^M$ for integers $a,b,c.$  These methods thus produce discriminants that are exponential in $M$, too large for practical purposes. 

The one exception where classical computation can be used to find class numbers with a large known divisor is when the divisor $M = 2^k.$
Bosma and Stevenhagen \cite{BosmaStevenhagen_2sylow} give an algorithm, formalizing methods described by Gauss \cite{Gauss} and Shanks \cite{Shanks_GATESR}, to compute the 2-Sylow subgroup of the class group of a quadratic field.  In addition to describing an algorithm that works in full generality, they prove that the algorithm runs  in expected time polynomial in $\log (|\Delta|)$.  Using this algorithm would enable an adversary to use trial-and-error efficiently to generate random primes $p$ until a sufficiently large power of $2$ divides the class number.

The primes $p$ recommended for use with CSIDH are not amenable to this method, because they are congruent to $3 \bmod 4,$ guaranteeing that the class number of the non-maximal order of discriminant $-4p$ is odd.  However, in Section~4 of \cite{CSIDH}, the authors write that they pick $p \equiv 3 \pmod{4}$ because it makes it easy to write down a supersingular curve, but that ``in principle, this constraint is not necessary for the theory to work".  We suggest that restricting to primes $p \equiv 3 \pmod{4}$ is in fact necessary, in order to avoid unnecessary potential vulnerabilities.

\section{Conclusion}

We described two variants of a  quantum algorithm for computing an isogeny between two elliptic curves $E_1,E_2$ 
defined over a finite field such that there is an imaginary quadratic order $\OO$ 
satisfying $\OO\simeq \End(E_i)$ for $i = 1,2$ with $\Delta = \disc(\OO)$. Our first variant runs in 
in heuristic asymptotic run time $2^{O\left(\sqrt{\log(|\Delta|)}\right)}$ and 
requires polynomial quantum memory and $2^{O\left(\sqrt{\log(|\Delta|)}\right)}$ quantumly accessible classical memory. 
The second variant of our algorithm relying on Regev's 
dihedral HSP solver~\cite{regev} runs in time $e^{\left(\frac{1}{\sqrt{2}}+ o(1)\right)\sqrt{\ln(|\Delta|)\ln\ln(|\Delta|)}}$ while relying only on 
polynomial (classical and quantum) memory. These variants of the HSP-based algorithms for computing isogenies have the best asymptotic 
complexity, but we left the assessment of their actual cost on specific instances such as the proposed CSIDH parameters~\cite{CSIDH} for 
future work. Some of the constants involved in lattice reduction were not calculated, and more importantly, the role of the memory requirement 
should be addressed in light of the recent results on the topic~\cite{Isogeny_memory_cost}. 

\section*{Acknowledgments}

The authors thank L\'{e}o Ducas for useful comments on the memory requirements of the BKZ algorithm. The authors also thank Tanja Lange and Benjamin Smith for useful comments 
on an earlier version of this draft. 

%
%
\bibliographystyle{plain}
\bibliography{biblio}

\end{document}